\documentclass[twoside,journal]{IEEEtran}
\usepackage{amsfonts}
\usepackage{amssymb}
\usepackage{amsmath}
\usepackage{mathrsfs}
\usepackage{verbatim}
\usepackage{subfigure}
\usepackage{balance}
\usepackage{booktabs}
\usepackage{fancybox}
\usepackage{bm}
\usepackage{changepage}
\usepackage{extarrows}
\usepackage{algorithm}
\usepackage{algorithmic}
\usepackage{diagbox}
\usepackage{multirow}
\usepackage{array}
\usepackage{graphicx}
\usepackage{epstopdf}
\usepackage{caption}
\usepackage{color}
\newtheorem{theorem}{Theorem}
\newtheorem{lemma}{Lemma}
\newtheorem{corollary}{Corollary}
\newtheorem{proof}{Proof}
\newtheorem{proposition}{Proposition}

\newtheorem{remark}{Remark}

\begin{document}

	\title{Multi-Antenna Covert Communications {in Random Wireless Networks}}
\author{Tong-Xing~Zheng,~\IEEEmembership{Member,~IEEE,}
Hui-Ming~Wang,~\IEEEmembership{Senior~Member,~IEEE,}\\
Derrick~Wing~Kwan~Ng,~\IEEEmembership{Senior~Member,~IEEE,}
and~Jinhong~Yuan,~\IEEEmembership{Fellow,~IEEE}
\thanks{T.-X.~Zheng and H.-M. Wang are with the School of Electronic and Information Engineering, Xi'an Jiaotong University, Xi'an 710049, China
(e-mail: zhengtx@mail.xjtu.edu.cn, xjbswhm@gmail.com).}
\thanks{D. W. K. Ng and J. Yuan are with the School of Electrical Engineering and Telecommunications, University of New South Wales, Sydney, NSW 2052, Australia (e-mail: w.k.ng@unsw.edu.au, j.yuan@unsw.edu.au).}
}                         
\maketitle
\vspace{-2.1 cm}
	
\begin{abstract}
This paper studies multi-antenna-aided covert communications coexisting with randomly located wardens and interferers, considering both centralized and distributed antenna systems (CAS/DAS).
The throughput performance of the covert communication is analyzed and optimized under a stochastic geometry framework, where the joint impact of the small-scale channel fading and the large-scale path loss is examined.
To be specific, two probabilistic metrics, namely, the covert outage probability and the connectivity probability, are {adopted} to characterize the covertness and reliability of the transmission, respectively, and analytically tractable expressions for the two metrics are derived. 
The worst-case covert communication scenario is then investigated, {where the wardens invariably can maximize the covert outage probability by adjusting the detection thresholds for their detectors}. 
Afterwards, the optimal transmit power and transmission rate are jointly designed to maximize the covert throughput subject to a covertness constraint. 
Interestingly, it is found that the maximal covert throughput for both the CAS and DAS is invariant to the density of interferers and the interfering power, regardless of the number of transmit antennas. 
Numerical results demonstrate that the CAS outperforms the DAS in terms of the covert throughput for the random network of interest, and the throughput gap between the two systems increases dramatically when the number of transmit antennas becomes larger.
\end{abstract}

	
\IEEEpeerreviewmaketitle
	
\section{Introduction}

\IEEEPARstart{I}{n} the era of {Internet-of-Things} (IoT), the provisioning of security and privacy has become a critical issue due to a soaring amount of {devices communicating} confidential and sensitive information, e.g., financial details, identity authentication, and medical records, etc, over the open wireless media \cite{Bloch2011Physical}. Various security methods through cryptographic encryption \cite{Menezes1996Handbook,Talbot2006Complexity} or physical-layer (information-theoretic) security  \cite{Yang2015Safeguading}-\cite{Wang2016Physical_book} have been developed to prevent the message content from being intercepted by unintended recipients. Nevertheless, there are many real-life circumstances where safeguarding content secrecy is far from sufficient, and the communicating parties may desire to transfer the message covertly.
Typical examples include hiding military operations to keep from being detected by enemies, or concealing secret activities of an organization to escape the attention of an authoritarian government monitoring the network. Against this background, \emph{covert communication}, or termed \emph{low probability of detection} (LPD) communication, which aims to hide the very existence of the communication itself from watchful adversaries, has recently drawn considerable research interests \cite{Bash2013Limits}-\cite{Abdelaziz2017Fundamental}.

\subsection{Previous Works and Motivations}

Since the early 20th century, spread-spectrum techniques have been extensively applied for achieving covert communications, particularly for military applications \cite{Simon1994Spread}. Nonetheless, the fundamental information-theoretic limits of covert communications have not been explored until recently. Specifically, a \emph{square root law} was presented in \cite{Bash2013Limits} for additive white Gaussian noise (AWGN) channels, which states that in $n$ channel uses, at most ${O}(\sqrt{n})$ bits of information can be conveyed to an intended receiver reliably and covertly against a vigilant adversary (warden Willie). 
This seminal work was later extended to various channel models such as binary symmetric channels \cite{Che2013Reliable}, discrete memoryless channels \cite{Bloch2016Covert,Wang2016Fundamental}, multiple access channels \cite{Arumugam2016Keyless}, and multi-input multi-output AWGN channels \cite{Abdelaziz2017Fundamental}.

It is worth mentioning that the square root law built in \cite{Bash2013Limits} manifests that the achievable covert rate, {i.e., the rate at which reliability and covertness are guaranteed simultaneously}, approaches zero as $n$ grows to infinity, i.e., $\lim_{n\rightarrow\infty}\frac{{O}(\sqrt{n})}{n}=0$. 
Such a pessimistic conclusion motivates
increasing endeavors to be devoted to exploring the condition in which a positive covert rate can be promised. Fortunately, it has been proven that a positive covert rate {is still achievable} when the warden has various uncertainties in terms of the receiver noise power \cite{Lee2015Achieving}-\cite{He2017On}, the exact {timing of} the covert communication \cite{Bash2014LPD,Bash2016Covert}, the fading channel \cite{Shahzad2017Covert}, and the jamming signal deliberately emitted either by the destination itself \cite{Hu2017Covert} or by an external friendly helper \cite{Sobers2017Covert,Soltani2017Covert}. 
Furthermore, a recent work \cite{He2017Covert} showed that the ambient signals from coexisting interferers also can be exploited to produce a positive covert rate. By modeling the interferers' positions as a Poisson point process (PPP) \cite{Haenggi2009Stochastic}, the authors in \cite{He2017Covert} revealed that the maximal covert rate for the interference-limited network is invariant to the density of interferers. 

The vast majority of existing literature concerning covert communications has been focused on a single-antenna transmitter \cite{Lee2015Achieving}-\cite{He2017Covert}, whereas the multi-antenna-assisted covert communication has not been well investigated. 
Multi-antenna communication architectures are categorized into centralized antenna systems (CASs) and distributed antenna systems (DASs).
{In the CAS, the antennas are co-located on a single device, and a joint signaling design among the antennas can significantly boost the spectrum efficiency. In the DAS, the antennas are geographically spread and connected to a central processor using coax cable or optical fiber. Compared with the CAS, the DAS can provide rich spatial diversity to combat path loss and shadowing, reduce the average distance between a transmit antenna and a receiver, and create more uniform coverage \cite{Heath2013A}.} 
Both the CAS and DAS have been substantially examined in the context of physical-layer secure transmissions and have been shown to gain a remarkable security enhancement for various wireless networks  \cite{Zhang2013Enhancing}-\cite{Wang2016Artificial}.
When applying multiple antennas to covert communications, two fundamental questions are naturally raised: \emph{ 1) {How} multi-antenna techniques benefit covert communications?; 2) Which multi-antenna architecture is more applicable to covert communications?} Theoretically, a multi-antenna transmitter is capable to use less power to support a reliable transmission by adequately exploiting the spatial degrees of freedom.
A lower energy leakage in return embarrasses the detection for a warden. In this sense, multi-antenna techniques, if designed properly, can be {{beneficial}} for covert communications. For another thing, the {{co-located}} antennas for the CAS release a higher power to a neighboring warden compared with any of the distributed antennas for the DAS because of the spatial energy dispersion in the latter.
However, if a warden is likely to appear anywhere in a network and meanwhile its location is uninformed, the geographically spread antennas actually take a higher risk of being detected by the warden. This might even offset the potential benefit brought by the spatial energy dispersion. 
Therefore, it is not intuitive whether the CAS or DAS is better suited to covert communications.

Yet so far the two questions posed above have not been answered explicitly, and the potential of multi-antenna techniques for covert communications in fading channels has not yet been excavated. 
In particular, whether the CAS or DAS can provide a higher covert communication rate {is unclear} and the performance gap between them still remain unknown. 
Moreover, existing literature on covert communications has rarely taken into account of multiple wardens and the uncertainty of their spatial locations when designing the covert communication. In practice, there exist situations where wardens desire to hide themselves for a covert detection, and then their locations appear to be random to the monitored entity. Although the authors in a recent work \cite{Soltani2017Covert} considered multiple randomly distributed wardens, they only concentrated on the single-antenna system for AWGN channels, {and their results are not applicable to multi-antenna systems with fading channels}. 
All the shortcomings mentioned above  motivate the current research work.

\subsection{Contributions}

This paper explores the covert communication for a random network where a multi-antenna transmitter communicates with a single-antenna receiver against randomly distributed single-antenna wardens and interferers.
A comprehensive analysis and optimization framework for the covert throughput of the system is provided. In particular, tools from the stochastic geometry theory \cite{Haenggi2009Stochastic} are used to capture the impact of channel fading and path loss on the system performance.
The main contributions of this paper are summarized as follows.

\begin{itemize}
	\item The covert communication for both the CAS and DAS is investigated, where maximal ratio transmitting (MRT) and distributed beamforming (DBF) are employed as transmit strategies, respectively. For each multi-antenna system, analytical expressions for the covert outage probability and the connectivity probability are derived, where the two metrics are used to depict the covertness and reliability of the covert communication, respectively. 
\item

An optimization framework incorporating the designs in terms of the detection of wardens and transmission parameters is established. 
Specifically, the worst-case scenario of the covert communication is examined in which the optimal detection thresholds are determined from the perspective of wardens. Subsequently, a maximal covert throughput is achieved through a joint optimization of the transmit power and the transmission rate. 
		
\item
Various useful insights into the multi-antenna covert communication are provided. In particular, an \emph{invariance} property is revealed for both the CAS and DAS, which states that the maximal covert throughput for an interference-limited system is invariant to either the density of interferers or the interfering power. It is also demonstrated that the CAS always reaps a throughput gain over the DAS, and the gain enlarges with more transmit antennas. 

\end{itemize}


%
\subsection{Organization and Notations}
The remainder of this paper is organized as follows. 
Section II details the system model. 
Sections III and IV analyze and optimize comprehensively the covert throughput for the CAS and DAS, respectively. 
Section V presents numerical results to validate the theoretical analyses.
Section VI draws a conclusion of this paper.

\emph{Notations}: Bold lowercase letters denote column vectors. $|\cdot|$, $\|\cdot\|$, $(\cdot)^{\dagger}$, $(\cdot)^{\rm T}$, $(\cdot)^{\rm H}$, $\ln(\cdot)$, $\mathbb{P}\{\cdot\}$, $\mathbb{E}_v[\cdot]$ denote the absolute value, Euclidean norm, conjugate, transpose, Hermitian
transpose, natural logarithm, probability, and the expectation taken over a random variable $v$, respectively.
$f_v(\cdot)$ and $\mathcal{F}_v(\cdot)$ denote the probability density function (PDF) and the cumulative distribution function (CDF) of $v$, respectively.	
$\mathbb{L}_x$ denotes the polar coordinate $(r_{x,o},\theta_{x,o})$ with a distance $r_{x,o}$ and an angle $\theta_{x,o}$ to the origin $o$. $\mathcal{B}(o,D)$ denotes the disc centered at $o$ with a radius $D$. 

\section{System Model}
\begin{figure}[!t]
	\centering
	\includegraphics[width = 3.6in]{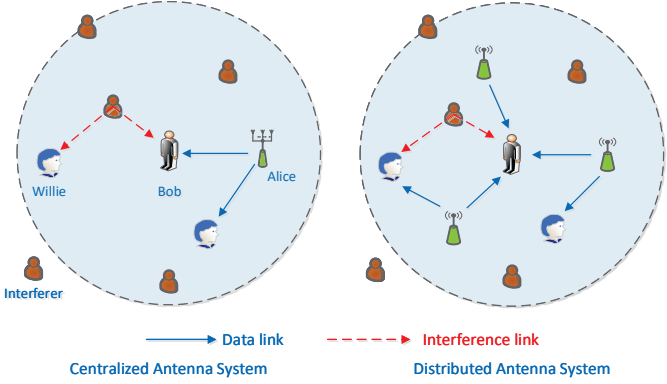}
	\caption{Illustration of a multi-antenna covert communication system. 
	{Alice's transmit antennas are co-located for the CAS and are deployed at different places for the DAS (three antennas in both figures).} Multiple wardens Willies move randomly and independently inside a certain region around Bob, and they aim to detect any transmission by Alice (two wardens within the dashed circle in the figure). There coexist numerous interferers randomly located in the network (five interferers in the figure). }
	\label{System_model}
\end{figure}
Consider a two-dimensional wireless network comprised of a source Alice, a destination Bob, $N$ wardens Willies, and numerous interferers, as depicted in Fig. \ref{System_model}. 
Wardens seek to detect any transmission by Alice, and Alice desires to deliver messages to Bob reliably while guaranteeing a low probability of being detected by Willies.
Alice is equipped with $M$ antennas while all the other nodes including Bob, Willies, and the interferers each are single-antenna devices. 
Consider two different multi-antenna paradigms, namely, the CAS and DAS, where Alice's transmit antennas are deployed together and are dispersed geographically, respectively. 
Without loss of generality, Bob is placed at the origin $o$ of the polar coordinate and the location of Alice's $m$-th antenna is denoted as $\mathbb{L}_{a_m}$. In particular, the $M$ antennas for the CAS share the same position $\mathbb{L}_{a}$. 
Suppose that the $N$ Willies are located independently and uniformly inside a disc $\mathcal{B}(o,D)$ centered with Bob such that the distribution of their locations $\{\mathbb{L}_{w}\}$ follows a binomial point process (BPP) $\Phi_W$ within $\mathcal{B}(o,D)$, i.e., $\mathbb{L}_{w}\in\Phi_W$.
The interferers are assumed to be scattered randomly in the network and their locations $\{\mathbb{L}_{j}\}$ are modeled as a homogeneous PPP $\Phi_J$ with density $\lambda_J$ on the entire two-dimensional plane, i.e., $\mathbb{L}_{j}\in\Phi_J$ \cite{Zhang2013Enhancing}-\cite{Wang2016Physical}.

\subsection{Channel Model}

All the wireless channels in the network undergo a standard distance-based path loss along with a frequency flat Rayleigh fading. 
The complex channel gains from Alice's $m$-th transmit antenna and from the interferer located at $\mathbb{L}_j$ to a receiving node at $\mathbb{L}_x$ are respectively expressed by $h_{a_m,x}r_{a_m,x}^{-\alpha/2}$ and $h_{j,x}r_{j,x}^{-\alpha/2}$, where $h_{a_m,x}$, $h_{j,x}$ denote the fading coefficients, $r_{a_m,x}$, $r_{j,x}$ denote the corresponding distances, and $\alpha$ denotes the path-loss exponent. 
For convenience, define $\bm h_{a,x}\triangleq \left[h_{a_1,x},\cdots,h_{a_M,x}\right]^{\rm T}$ as the fading coefficient vector from Alice to the receiver at $\mathbb{L}_x$.

Consider a time-slotted system where the locations of all the nodes and the fading coefficients remain static in a time slot. 
Assume that Alice knows perfectly the instantaneous channel state information of the channel from herself to Bob, i.e., $\bm h_{a,o}$. Hence, she can adapt the weight coefficients for her antennas to boost the received signal strength for Bob. Specifically, Alice employs MRT and DBF as transmit strategies for the CAS and DAS, respectively, where the weight coefficient for the $m$-th transmit antenna is devised in the form of\footnote{{Artificial noise is commonly exploited to confuse eavesdroppers in literature on physical-layer security, e.g., \cite{Zheng2015Multi}-\cite{Deng2016Artificial}, which however is not considered for covert communications here. The reason behind is twofold: allocating part of transmit power for artificial noise will barely change Willie's total received power so that it can scarcely improve the covertness due to the energy detection at Willie side. On the contrary, it will lower the power available for message delivery, thus degrading the reliability. }} \begin{align}\label{weight}
g_m &
=\begin{cases}
\frac{h_{a_m,o}^{\dagger}}{\|\bm h_{a,o}\|},& \rm CAS,\\
\frac{h_{a_m,o}^{\dagger}}{|h_{a_m,o}|},& \rm DAS.
\end{cases}
\end{align}
When Alice transmits a symbol $s[k]$, the signals received by Bob and by Willie at $\mathbb{L}_w$, denoted as $y_{o}[k]$ and $y_{w}[k]$, respectively, are uniformly expressed as
\begin{equation}\label{y_x}
y_x[k] = U_x[k]+V_x[k]+z_x[k],~ x\in\{o,w\},
\end{equation}
where $U_x[k]=\sum_{m=1}^M\sqrt{P_{m}}g_m h_{a_m,x}r_{a_m,x}^{-\alpha/2}s[k]$ is the total signal collected from Alice, with $P_m$ being the transmit power of Alice's $m$-th antenna; $V_x[k] = \sum_{\mathbb{L}_j\in\Phi_J}{\sqrt{P_J}h_{j,x}r_{j,x}^{-\alpha/2}}v_{j}[k]$ is the aggregate received interference, with $u_{j}[k]$ and $P_J$ being the signal radiated from the interferer at $\mathbb{L}_j$ and its transmit power, respectively; $z_x[k]$ is the thermal noise at the receiver with variance $\sigma_x^2$. 
It is assumed that the transmitted signals $s[k]$, $v_{j}[k]$ and the fading coefficients $h_{a_m,x}$, $h_{j,x}$ are independent and identically distributed (i.i.d.) with zero mean and unit variance.

\subsection{Detection of Covert Communications}
Wardens Willies attempt to judge whether Alice is transmitting or not by performing an optimal statistical hypothesis test (such as the Neyman-Pearson test \cite{Lehmann2005Testing}) on the observed sequence $\{y_w[k]\}_{k=1}^K$ in a communication slot. To this end, Willies should distinguish two hypotheses, namely, the null hypothesis $\mathcal{H}_0$ meaning that Alice is not transmitting and the alternate hypothesis $\mathcal{H}_1$ indicating an ongoing communication. The two hypotheses are detailed as below:
\begin{align}
\mathcal{H}_0:~& y_{w}[k] = V_{w}[k]+z_{w}[k],\\
\mathcal{H}_1:~& y_{w}[k] = U_{w}[k]+V_{w}[k]+z_{w}[k].
\end{align}
Willies' ultimate goal is to detect whether $\{y_w[k]\}_{k=1}^K$ comes from $\mathcal{H}_0$ or $\mathcal{H}_1$.
{A correct detection corresponds to either the acceptance of $\mathcal{H}_0$ when it is true or the rejection of $\mathcal{H}_0$ when it is false. 
The probability of a correct detection is termed the \emph{detection probability}.}
{Given that a radiometer  is generally employed in practice for detection, it is assumed here that Willies adopt the radiometers as their detectors as well. This assumption is justified in \cite{He2017Covert}.}
With a radiometer, the decision rule is described as below:
\begin{equation}\label{test_rule}
\bar{P}_{w}\underset{\mathcal{D}_0}{\overset{\mathcal{D}_1}{\gtrless}} \xi,
\end{equation}
where the test
statistic $\bar P_{w}=\frac{1}{K}\sum_{k=1}^K| y_{w}[k]|^2$ is given by the average power received by Willie at $\mathbb{L}_w$ in a time slot, and $\xi>0$ is a predefined detection threshold for the detector. $\mathcal{D}_0$ and $\mathcal{D}_1$ stand for the binary decisions in favor of $\mathcal{H}_0$ and $\mathcal{H}_1$, respectively, and decision $\mathcal{D}_0$ is made if $\bar P_{w}\le \xi$ whereas decision $\mathcal{D}_1$ is made otherwise. 
{In this way, the detection probability for Willie at $\mathbb{L}_w$ is defined as $p_w \triangleq \mathbb{P}\left[\mathcal{H}_0\right]\mathbb{P}\left[\mathcal{D}_0|\mathcal{H}_0\right]+\mathbb{P}\left[\mathcal{H}_1\right]\mathbb{P}\left[\mathcal{D}_1|\mathcal{H}_1\right]$, where $p_{w}=1$
	corresponds to a perfect detection and on the contrary $p_{w}=0.5$ is no better than random guessing. For simplicity, consider equal \emph{a priori} probabilities of hypotheses $\mathcal{H}_0$ and $\mathcal{H}_1$ such that $\mathbb{P}\left[\mathcal{H}_0\right]=\mathbb{P}\left[\mathcal{H}_1\right]=0.5$.}

The detection probability is affected by the uncertainties from transmitted signals, receiver noise, fading channels, and node positions.
Assume that Willies exploit an infinite number of signal samples to perform the detection, i.e., $K\rightarrow\infty$, then the uncertainties of transmitted signals and receiver noise vanish. Consequently, the average received power $\bar P_{w}$ is rewritten as  
\begin{align}\label{PW}
\bar P_{w} &
=\begin{cases}
I_{w}+\sigma_{w}^2,& \mathcal{H}_0,\\
S_{w}+I_{w}+\sigma_{w}^2,& \mathcal{H}_1 ,
\end{cases}
\end{align}
where $S_{w} = \big|\sum_{m=1}^M\sqrt{P_{m}}g_m h_{a_m,w}r_{a_m,w}^{-\alpha/2}\big|^2 $ and $I_{w}=\sum_{\mathbb{L}_j\in\Phi_J}{P_J}|h_{j,x}|^2r_{j,x}^{-\alpha}$ denote the received signal power from Alice and from the interferers, respectively. 
Recalling the decision rule described in \eqref{test_rule}, 
{the detection probability $p_{w}$ for certain channel realizations and node locations (i.e., for given $S_w$, $I_w$, and $\sigma_w^2$) is calculated as below:}
\begin{align}\label{de}
p_{w} 
=\begin{cases}
1,& I_{w}+\sigma_{w}^2\le\xi< S_{w}+I_{w}+\sigma_{w}^2,\\
0.5,& \rm otherwise.
\end{cases}
\end{align}
{Note that the detection probability $p_{w}$ is either 1 or 0.5, depending on the setting of detection threshold $\xi$. Moreover, if taking into consideration the randomness of $S_{w}$ and $I_{w}$, the detection probability $p_{w}$ would become a Bernoulli distributed random variable for any fixed threshold $\xi$.} 

\subsection{Performance Metrics}
This subsection introduces several metrics which are used to characterize the performance of the covert communication system under investigation. 
\subsubsection{Covert Outage Probability}
The covert communication between Alice and Bob fails when it is detected by any Willie, and then a covert outage event is said to have occurred. The probability that this event happens is referred to as the \emph{covert outage probability} \cite{He2017On}, denoted as $\mathcal{O}$, which quantifies the covertness of the communication.
{{Since the detection probability $p_{w}$ given in \eqref{de} is a Bernoulli random variable, the covert outage probability $\mathcal{O}$ is defined as the probability that there is at least one Willie having a detection probability equal to one, i.e.,}}
\begin{equation}\label{def_oco}
\mathcal{O} = \mathbb{E}_{\Phi_W}\left[\mathbb{P}\left\{\bigcup_{\mathbb{L}_w\in\Phi_W}p_{w}=1\right\}\right].
\end{equation}
Note that the inner probability in \eqref{def_oco} is operated over the random variables $S_w$ and $I_w$ for $\mathbb{L}_w\in\Phi_W$ and the outer expectation is taken over Willies' random locations $\{\mathbb{L}_w\}$.

\subsubsection{Connectivity Probability}
Revisit the received signal $y_{o}[k]$ in \eqref{y_x}, and the  signal-to-interference-plus-noise ratio (SINR) of the channel from Alice to Bob is expressed as \begin{equation}\label{sinr}
\gamma_o = \frac{S_o}{I_o+\sigma_o^2},
\end{equation}
where $S_{o} = \big|\sum_{m=1}^M\sqrt{P_{m}}g_m h_{a_m,o}r_{a_m,o}^{-\alpha/2}\big|^2  $ and $I_{o}=\sum_{\mathbb{L}_j\in\Phi_J}{P_J}|h_{j,o}|^2r_{j,o}^{-\alpha}$ denote the power of the desired signal from Alice and the aggregate interference power, respectively. 
With \eqref{sinr}, the achievable rate of Bob is given by $C_o = \ln (1+\gamma_o)$ nats/s/Hz.
If a target transmission rate $R$ can be supported, i.e., $C_o\ge R$, Alice is deemed to be successfully connected to Bob, and Bob can recover the messages delivered from Alice. The metric \emph{connectivity probability}, denoted as $\mathcal{C}$, is adopted to measure transmission reliability, and is defined as the probability that the SINR $\gamma_o$ is larger than or equal to the SINR threshold $\beta\triangleq e^R-1$, as given below:
\begin{equation}\label{def_oc}
\mathcal{C} = \mathbb{P}\left\{\frac{S_o}{I_o+\sigma_o^2}\ge\beta\right\}.
\end{equation}

\subsubsection{Covert Throughput}
A core metric named \emph{covert throughput}, denoted as $\mathcal{T}$, is employed in order to evaluate the rate efficiency of the covert communication. The covert throughput is defined as the average successfully transmitted amount of information per second per Hertz subject to a covertness requirement $\mathcal{O}\le\epsilon$, where the threshold $\epsilon\in[0,1]$ represents the maximal acceptable covert outage probability. 
{Formally, the covert throughput is expressed as the product of the connectivity probability $\mathcal{C}$ and the transmission rate $R$, which is described as
\begin{equation}\label{ct}
\mathcal{T} = \mathcal{C}R, ~~ \mathcal{O}\le\epsilon.
\end{equation}
If the covertness constraint $\mathcal{O}\le\epsilon$ is violated, $\mathcal{T}$ is set to zero.}
The covert throughput defined in \eqref{ct} would turn into the well-known secrecy throughput, if the covertness constraint changes to a secrecy outage probability constraint \cite{Zheng2017Physical}.

The following two sections proceed to the covert throughput maximization for the CAS and DAS, respectively. Due to uncoordinated concurrent transmissions by the interferers, the aggregate interference power at a receiver typically dominates the noise power. For tractability, an interference-limited network is focused on by ignoring the thermal noise such that both $\sigma_w^2$ in \eqref{de} and $\sigma_o^2$ in \eqref{def_oc} are removed. In fact, the obtained results can
be easily generalized to the case with the inclusion of the thermal noise, which however would only complicate the analysis but provide no significant qualitative difference.

\section{Centralized Antenna Systems}
This section examines the covert communication for the CAS, where Alice places all her antennas at the same location $\mathbb{L}_a$, i.e., $(r_{a,o},\theta_{a,o})$. The total transmit power of Alice is denoted as $P_A$. Before proceeding to maximizing the covert throughput, some important insights into the covert outage probability $\mathcal{O}$ and the connectivity probability $\mathcal{C}$ are provided. 

\subsection{Covert Outage Probability}
Based on the detection probability $p_{w}$ in \eqref{de}, the covert outage probability $\mathcal{O}$ defined in \eqref{def_oco} can be interpreted as the complement of the probability that all the Willies' detection probabilities are less than one, which is reformulated as
 \begin{align}\label{oco}
\mathcal{O} &=1 -  \mathbb{E}_{\Phi_W}\left[\mathbb{P}\left\{\bigcap_{\mathbb{L}_w\in\Phi_W}p_{w}<1\right\}\right]\nonumber\\
&= 1 -  \mathbb{E}_{\Phi_W}\left[\prod_{\mathbb{L}_w\in\Phi_W}\left(1-\bar p_{w}\right)\right],
\end{align}
where $\bar p_{w}\triangleq \mathbb{P}\left\{p_{w}=1\right\}$ denotes the average detection probability for Willie at $\mathbb{L}_w\in\Phi_W$. 
{Note that the second equality in \eqref{oco} follows from the assumption that Willies do not collude with each other such that their detections are independent.}
Before computing $\mathcal{O}$, it is needed to calculate $\bar p_{w}$, which can be obtained from \eqref{de} by averaging over the random variables $S_w$ and $I_w$, i.e.,
\begin{align}\label{average_pe}
\bar p_{w}&= \mathbb{P}\{I_w\le\xi<S_w+I_w\}=\mathcal{F}_{I_w}(\xi)-\mathcal{F}_{S_w+I_w}(\xi)\nonumber\\
&=  \mathcal{F}_{I_w}(\xi)-\int_0^{\xi}\mathcal{F}_{I_w}(\xi-x)f_{S_w}(x)dx,\!\!
\end{align}
where $S_w = P_A\frac{|\bm h_{a,o}^{\rm H} \bm h_{a,w}|^2}{\|\bm h_{a,o}\|^2}r_{a,w}^{-\alpha}$.
It is verified that $S_w$ is exponentially distributed with $f_{S_w}(x) = \frac{r_{a,w}^{\alpha}}{P_A}e^{-{r_{a,w}^{\alpha}x}/{P_A}}$ \cite{Zheng2015Multi} which is independent of $M$. In other words, there is no statistical difference from the performance of Willie whether Alice uses a single antenna or multiple antennas when Alice adopts MRT. Hence, adding transmit antennas will exert no impact on the covert outage probability as long as the total transmit power is fixed. As the interferers' locations are modeled by a homogeneous PPP, the aggregate interference $I_w$ is the shot noise \cite{Lowen1990Power}. 
Generally, $\mathcal{F}_{I_w}(x)$ only can be displayed in an infinite series \cite{Lowen1990Power}, which causes a high computational complexity to calculate $\bar p_{w}$. In order to mitigate the calculation burden, the Laplace transform of $I_w$ is invoked.
\begin{lemma}[{\cite[Eqn. (8)]{Haenggi2009Stochastic}}]\label{lemma_laplace_i}
The Laplace transform of $I_w$ evaluated at $s$ is given by
\begin{equation}\label{laplace_i}
\mathcal{L}_{I_w}(s)=\mathbb{E}_{I_w}\left[e^{-sI_w}\right] =e^{-\kappa\lambda_J P_J^{\delta} s^{\delta}},
\end{equation}
where $\delta = 2/\alpha$, $\kappa\triangleq \pi\Gamma(1+\delta)\Gamma(1-\delta)$, and $\Gamma(z)$ is the gamma function \cite[Eqn. (8.310.1)]{Gradshteyn2007Table}.
\end{lemma}

With the aid of Lemma \ref{lemma_laplace_i} together with a widely used approximation approach \cite{Singh2015Tractable}, a closed-form expression for $\mathcal{F}_{I_w}(x)$ is provided by the following lemma.
\begin{lemma}\label{lemma_cdf_i}
The CDF $\mathcal{F}_{I_w}(x)$ is approximated by
\begin{equation}\label{cdf_i}
{{\mathcal{F}}_{I_w}(x) \approx} \sum_{l=1}^L {L\choose l}(-1)^{l+1}e^{-\kappa\lambda_J P_J^{\delta}\left(\frac{l\tau }{x}\right)^{\delta}},
	\end{equation}
	where $\tau \triangleq L(L!)^{-1/L}$ and $L$ is the number of terms applied for the approximation.
\end{lemma}
{\begin{proof}
	$\mathcal{F}_{I_w}(x)$ is calculated as follows,
		\begin{align}\label{cdf_iw}
		\mathcal{F}_{I_w}(x)  &= \mathbb{P}\left\{I_w\le x\right\} 
		=\mathbb{P}\left\{{I_w}/{x}\le 1 \right\}\stackrel{\mathrm{(a)}}\approx
		\mathbb{P}\left\{{I_w}/{x}\le \iota \right\}\nonumber\\
		&	\stackrel{\mathrm{(b)}}\lessapprox 1 - \mathbb{E}_{I_w}\left[\left(1 - e^{-{\tau I_w}/{x}}\right)^L\right],
		\end{align}
		where the dummy variable $\iota$ introduced in $\rm (a)$ is a normalized gamma random variable with the shape parameter $L$, and $\rm (a)$ follows from the fact that $\iota$ converges to one as $L$ approaches infinity \cite{Singh2015Tractable}; $\rm (b)$ yields a tight upper bound by invoking Alzer's inequality \cite{Alzer1997On}, i.e., $\mathbb{P}\{\iota\ge z\}\lessapprox 1 - \left[1-e^{-\xi z}\right]^L$ for a constant $z>0$. Using the binomial expansion with \eqref{cdf_iw} and plugging the Laplace transform $\mathcal{L}_{I_w}\left(s\right)$ in \eqref{laplace_i} with $s={l\tau}/{x}$ completes the proof.
\end{proof}}

Substituting the approximated CDF $\mathcal{F}_{I_w}(x)$ into \eqref{average_pe}, the following theorem is obtained.
\begin{theorem}\label{theorem_average_pe}
	The average detection probability $\bar p_{w}$ in \eqref{average_pe} is approximated by
\begin{equation}\label{app_pe}
{\bar p_{w}\approx} \sum_{l=1}^L {L\choose l}(-1)^{l+1}
\left[e^{-\kappa\lambda_J P_J^{\delta}\left(\frac{l\tau }{\xi}\right)^{\delta}}
-e^{-\frac{r_{a,w}^{\alpha}\xi}{P_A}}{Z}_l(\xi,\alpha)\right],
\end{equation}
where ${Z}_l(\xi,\alpha)=\frac{r_{a,w}^{\alpha}}{P_A}\int_0^{\xi}e^{-\kappa\lambda_J P_J^{\delta}\left({l\tau}/{y}\right)^{\delta}+{r_{a,w}^{\alpha}y}/{P_A}}dy$ with $\tau$ and $L$ defined in Lemma \ref{lemma_cdf_i}.
\end{theorem}

The expression of the approximate $\bar p_{w}$ in \eqref{app_pe} is simple and practically closed-form which requires only the computation or lookup of a ${Z}_l(\xi,\alpha)$ value. 
For a special case with $\alpha=4$, a closed-form expression for the PDF of $I_w$ is found in \cite{Haenggi2012Stochastic}, which is rewritten below:
\begin{equation}\label{pdf_i}
f_{I_w}(x) = \frac{\pi^{3/2}\lambda_J\sqrt{P_J}}{4x^{3/2}}e^{-\frac{\pi^{4}\lambda_J^2{P_J}}{16x}}.
\end{equation}
With \eqref{pdf_i}, the CDF $\mathcal{F}_{I_w}(x)$ is simplified as
\begin{equation}\label{cdf_i4}
\mathcal{F}_{I_w}(x) = \int_0^x f_{I_w}(x)dx = 1 - {\rm erf}\left(\frac{\pi^{2}\lambda_J\sqrt{P_J}}{4\sqrt{x}}\right),
\end{equation}
where ${\rm erf}(z)$ is the error function \cite[Eqn. (8.250.1)]{Gradshteyn2007Table}. Plugging \eqref{cdf_i4} into \eqref{average_pe} yields 
\begin{equation}\label{pe_4}
{\bar p}_{w}^{\alpha=4} = e^{-B\xi}\left[1+\int_0^{\xi}{\rm erf}\left(\frac{A}{\sqrt{y}}\right)Be^{By}dy\right]-{\rm erf}\left(\frac{A}{\sqrt{\xi}}\right),
\end{equation}
where $A = {\pi^2\lambda_J\sqrt{P_J}}/{4}$ and $B = {r_{a,w}^{4}}/{P_A}$. The correctness of the exact $\bar p_w$ given in \eqref{pe_4} is confirmed by Monte-Carlo simulation results as shown in Fig. \ref{ADP}.\footnote{{The simulation results are obtained by using 100,000 trails. Each trial distributes $N_i$ interferers uniformly as a BPP inside a sufficiently large square area $S_j=[-L_j~L_j]^2$, where $N_i$ is a Poisson random variable with mean $4\lambda_JL_j^2$. All the channel fading coefficients are generated to be i.i.d. complex Gaussian with zero mean and unit variance. The received signal power $S_w$ and interference power $I_w$ are computed from \eqref{PW}. Finally, the average detection probability $\bar p_w$ is computed by counting the number of times the event $I_w\le\xi<S_w+I_w$ in \eqref{average_pe} happens.}}
Moreover, the approximations derived in \eqref{app_pe} coincide well with \eqref{pe_4} when $L = 5$ is chosen. 
\begin{figure}[!t]
	\centering
	\includegraphics[width = 3.8in]{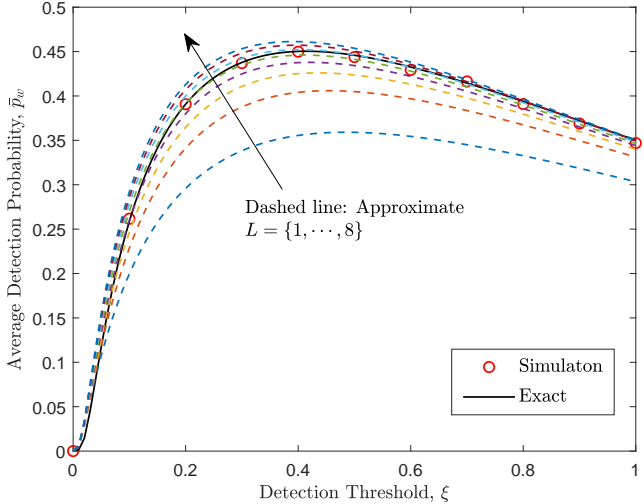}
	\caption{$\bar p_{w}$ vs. $\xi$ for different $L$'s, with $P_A=30$ dBm, $\lambda_J = 0.1$, and $r_{a,w} = 1$. Other parameters are specified in Sec. V.}
	\label{ADP}
\end{figure}

Having obtained ${\bar p}_{w}$, the covert outage probability $\mathcal{O}$ is derived in the following theorem.
\begin{theorem}\label{theorem_cop}
The covert outage probability for the CAS is given by
\begin{equation}\label{oco_expression}
\mathcal{O} = 1 - \left(\int_0^{2\pi}\int_0^D \left(1-{\bar p}_{w}\right) \frac{r}{\pi D^2}drd\theta\right)^N,
\end{equation}
where the parameter $r_{a,w}$ in ${\bar p}_{w}$ is expressed as $r_{a,w} = \sqrt{r_{a,o}^2+r^2 - 2r_{a,o}r\cos(\theta_{a,o} - \theta)}$.
\end{theorem}
\begin{proof}
Note that the distribution of the $N$ wardens' locations obeys a BPP. Due to the i.i.d. property of a BPP, the covert outage probability $\mathcal{O}$ given in \eqref{oco} is equivalently transformed to 
\begin{align}\label{oco_cal}
\mathcal{O} &= 1 -  \prod_{n=1}^N\mathbb{E}_{r_{a,w}}\left[1-{\bar p}_{w}\right]\nonumber\\
&= 1 - \prod_{n=1}^N\int_0^{2\pi}\int_0^D \left(1-{\bar p}_{w}\right) f_{r_{w,o},\theta_{w,o}}(r,\theta)drd\theta,
\end{align}
where $f_{r_{w,o},\theta_{w,o}}(r,\theta)= \frac{r}{\pi D^2}$ and substituting it into \eqref{oco_cal} completes the proof.
\end{proof}

{Although \eqref{oco_expression} involves two nested integrals, the integral interval is finite and thus the integral is practically not difficult to be numerically evaluated. }

\subsection{Optimal Detection Threshold from Willie's Viewpoint }
From a robust design perspective, the worst case of the covert communication between Alice and Bob is examined. In particular, the optimal detection threshold $\xi$ is designed from Willie's point of view which results in a maximal covert outage probability $\mathcal{O}$. In addition, the worst-case covert communication scenario should consider each Willie can adjust his detection threshold based on the distance between himself and Alice for improving detection accuracy. 

Since $\mathcal{O}$ in \eqref{oco_expression} increases with ${\bar p}_{w}$, to maximize $\mathcal{O}$ only requires to maximize ${\bar p}_{w}$ for each realization of Willie's location $\mathbb{L}_w$. 
{Theoretically, neither a too small nor a too large detection threshold $\xi$ is beneficial for detection, and there is an optimal $\xi$ that yields a maximal $\bar p_w$. 
Besides, this property should be  irrelevant to the path-loss exponent $\alpha$.
Hence, for mathematical tractability, only the special case with $\alpha=4$ is considered, and the following theorem provides the optimal detection threshold $\xi$ that maximizes ${\bar p}_{w}^{\alpha=4}$.}
\begin{theorem}\label{theorem_opt_xi}
	The average detection probability ${\bar p}_{w}^{\alpha=4}$ in \eqref{pe_4} initially increases and then decreases with the threshold $\xi$; the maximal ${\bar p}_{w}^{\alpha=4}$ is achieved at $\xi=\xi_o$ and is given by
	\begin{equation}\label{opt_pe}
	{\bar p}_{w,max}^{\alpha=4}={A}{B^{-1}}e^{-\frac{A^2}{\xi_o}}/\sqrt{\pi \xi_o^3},
	\end{equation} 
	where $\xi_o$ is the unique root of $\xi>0$ to the following equation,
	\begin{equation}\label{opt_xi}
{A}e^{-\frac{A^2}{\xi}}/{\sqrt{\pi \xi^3}}-Be^{-B\xi}Y(\xi)+B{\rm erf}({A}/{\sqrt{\xi}})= 0,
	\end{equation} 
with $A$ and $B$ defined in \eqref{pe_4}, and  $Y(\xi)=1+\int_0^{\xi}{\rm erf}\left({A}/{\sqrt{y}}\right)Be^{By}dy$.
\end{theorem}
\begin{proof}
	Please refer to Appendix \ref{appendix_theorem_opt_xi}.
\end{proof}

The first-increasing-then-decreasing trend of ${\bar p}_{w}$ with respect to (w.r.t.) $\xi$ is validated in Fig. \ref{ADP}.
Let $Y_1(\xi)$ denote the left-hand side of \eqref{opt_xi} such that $Y_1(\xi_o)=0$. As Appendix \ref{appendix_theorem_opt_xi} indicates, $Y_1(\xi)$ is first positive and then negative as $\xi$ grows from zero to infinity, then $\xi_o$ can be efficiently calculated via a bisection search with \eqref{opt_xi}.
Using the derivative rule for implicit functions \cite{Zheng2015Multi} with $Y_1(\xi_o)=0$ yields the derivative $\frac{d\xi_o}{dA}= -\frac{\partial Y_1(\xi_o)/\partial A}{\partial Y_1(\xi_o)/\partial \xi_o}>0$, namely, $\xi_o$ increases with $A$. 
Likewise, $\xi_o$ decreases with $B$. Since $A=\frac{\pi^2\lambda_J\sqrt{P_J}}{4}$ and $B = \frac{r_{a,w}^{4}}{P_A}$, it is inferred that $\xi_o$ increases with $\lambda_J$, $P_J$, and $P_A$. This suggests that Willie would enlarge the detection threshold when the interferer density, the interfering power, or Alice's transmit power increases, since only in this way can Willies distinguish Alice's signals from the interference more accurately.
It is easy to confirm that the above properties regarding $\xi_o$ for $\alpha=4$ are still valid for more general cases in which the optimal detection threshold can be obtained via an exhaustive search.

Although it is rather troublesome to exhibit $\bar p_{w,max}$ explicitly with $P_A$ due to the implicit form of $\xi_o$, the monotonicity of $\bar p_{w,max}$ w.r.t. $P_A$ is still revealed by the following corollary.
\begin{corollary}\label{corollary_pw_pa}
	The maximal average detection probability $\bar p_{w,max}$ for the worst-case covert communication monotonically increases with Alice's transmit power $P_A$.
\end{corollary}
\begin{proof}
	Consider Alice's transmit power $P_{A,1}$ and $P_{A,2}$ with $P_{A,2}>P_{A,1}$, and let $\xi_{o,1}$ and $\xi_{o,2}$ be the optimal detection thresholds maximizing $\bar p_w$ for $P_{A,1}$ and $P_{A,2}$, respectively. When $P_A$ increases, $S_w$ increases and then the feasible region of $\xi$, i.e., $[I_w,S_w+I_w)$, is enlarged.
	Hence, $\bar p_w(P_{A,2}, \xi_{o,2})\ge\bar p_w(P_{A,2}, \xi_{o,1})>\bar p_w(P_{A,1}, \xi_{o,1})$, where the second inequality holds since  $\bar p_w$ increases with $P_A$ for a constant $\xi$ as shown in \eqref{average_pe}. This completes the proof.
\end{proof}
 
Fig. \ref{OPT_ADP} illustrates how the average detection probability $\bar p_w$ is affected by the interferer density $\lambda_J$ and Alice's transmit power $P_A$. The optimal detection threshold $\xi_o$ is shown to significantly improve $\bar p_{w}$ compared with a constant $\xi$. 
It is observed that $\bar p_{w}$ increases with $P_A$ and decreases with $\lambda_J$ when the optimal $\xi_o$ is used for detectors. This demonstrates the harmfulness of high transmit power for covert communications, whereas the covertness performance indeed can be improved by introducing co-channel interference. Moreover, the monotonicity of the optimal $\xi_o$ w.r.t. $\lambda_J$ or $P_A$ is also confirmed in Fig. \ref{OPT_ADP} (see the circle dots in the figure). 

\begin{figure}[!t]
	\centering
	\includegraphics[width = 3.8in]{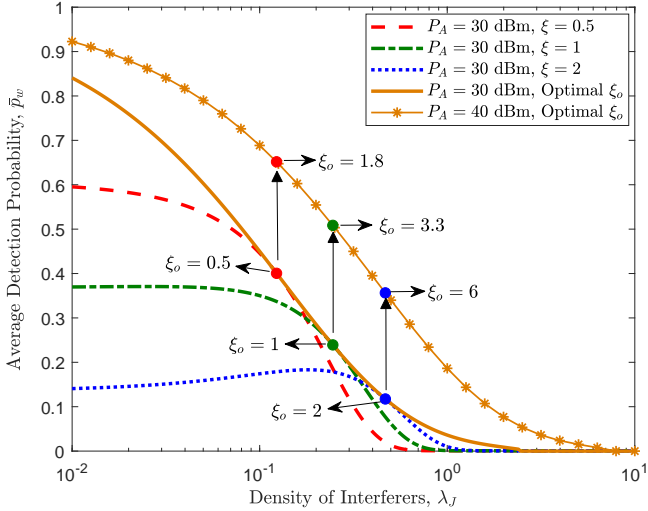}
	\caption{$\bar p_{w}$ vs. $\lambda_J$ for different $P_A$'s, with $r_{a,w} = 1$. }
	\label{OPT_ADP}
\end{figure}

\subsection{Connectivity Probability}
A closed-form expression for the exact connectivity probability $\mathcal{C}$ in \eqref{def_oc} is provide below.
\begin{theorem}\label{theorem_cp}
	The connectivity probability for the CAS is given by
	\begin{equation}\label{oc}
\mathcal{C} =  e^{-\phi\beta^{\delta}}
	+e^{-\phi\beta^{\delta}}\sum_{m=1}^{M-1}\frac{1}{m!}\sum_{n=1}^{m}\left(
	\delta\phi\beta^{\delta}\right)^n
	\Upsilon_{m,n},
	\end{equation}
	with
	$\Upsilon_{m,n}=\sum_{\psi_j\in \mathbb{C}\binom{m-1}{m-n}}\prod_
	{\substack{
	l_{ij}\in\psi_j\\
	i=1,\cdots,m-n}
	}
	\left[l_{ij}-\delta(l_{ij}-i+1)\right]$ and $\phi \triangleq {\kappa\lambda_J P_J^{\delta} r_{a,o}^2}/{P_A^{\delta}}$.
The term $\mathbb{C}\binom{m-1}{m-n}$ denotes the set of all distinct subsets $\psi_j$ of the natural numbers $\{1,2,\cdots,m-1\}$ with cardinality $m-n$.
	The elements in each subset are arranged in an ascending order with $l_{ij}$ being the $i$-th element of $\psi_j$.
	For $m\ge1$, it is set $\Upsilon_{m,m}=1$.
\end{theorem}
\begin{proof}
	Let $s\triangleq \frac{r_{a,o}^{\alpha}\beta}{P_A}$, and $\mathcal{C}$ in \eqref{def_oc} is calculated as follows:
	\begin{align}\label{oc_cal}
	\mathcal{C}&= \mathbb{E}_{I_o} \left[ \mathbb{P}\left\{S_o\geq \beta I_o\right\}\right]= \mathbb{E}_{I_o} \left[ \mathbb{P}\left\{\|\bm h_{a,o}\|^2\geq sI_o\right\}\right]
	\nonumber\\&\stackrel{\mathrm{(c)}}
	= \mathbb{E}_{I_o}\left[e^{-sI_o}
	\sum_{m=0}^{M-1}\frac{s^mI_o^m}
	{m!}\right] =\sum_{m=0}^{M-1}\mathbb{E}_{I_o} \left[
	\frac{s^me^{-sI_o}}
	{m!}I_o^m\right]\nonumber\\
	&
	\stackrel{\mathrm{(d)}}= \sum_{m=0}^{M-1}\left[
	\frac{(-s)^m}{m!}\frac{d^m \mathcal{L}_{I_o}(s)}{ds^m}\right],
	\end{align}
	where $\rm (c)$ holds since $\|\bm h_{a,o}\|^2$ is a normalized gamma random variable with the shape parameter $M$ and $\rm (d)$ is due to $\frac{(-1)^md^m \mathcal{L}_{I_o}(s)}{ds^m}=\mathbb{E}_{I_o} \left[
	I_o^me^{-sI_o}\right]$, where $\mathcal{L}_{I_o}(s)$ has the same form as $\mathcal{L}_{I_w}(s)$ in \eqref{laplace_i}. Substituting \eqref{laplace_i} into \eqref{oc_cal} and invoking \cite[Theorem 1]{Hunter08Transmission} complete the proof.
\end{proof}

{The correctness of Theorem \ref{theorem_cp} is validated by Fig. \ref{CP} in Sec. V.
The first part $e^{-\phi\beta^{\delta}}$ in $\mathcal{C}$ arises from a single antenna, and the second part is attributed to the deployment of multiple antennas. If adding one more antenna, $\mathcal{C}$ increases $\frac{1}{M!}\sum_{n=1}^{M}\left(
\delta\phi\beta^{\delta}\right)^n
\Upsilon_{M,n}$, but the increment becomes insignificant when $M$ is sufficiently large. This implies, it is unnecessary to employ antennas excessively, and a favorable reliability still can be achieved.} 

\subsection{Covert Throughput Maximization}
This subsection optimizes the covert throughput $\mathcal{T}=\mathcal{C}R$ subject to a covert outage probability constraint $\mathcal{O}\le\epsilon$.  The optimization problem is formulated as follows:
\begin{equation}\label{ct_max}
\max_{P_A>0,R>0}\mathcal{T}=\mathcal{C}R,
~{\rm s.t.} ~\mathcal{O}\le\epsilon.
\end{equation}
{Since $\mathcal{C}$ is a function of both $P_A$ and $R$ but $\mathcal{O}$ is independent of $R$, in order to maximize $\mathcal{T}$, the following equivalent transformation is carried on
according to \cite[Sec. 4.1.3]{Boyd2004Convex},
\begin{equation}\label{ct_max_equal}
\max_{P_A>0,R>0}\mathcal{C}(P_A,R)R=\max_{R>0}\left(R\max_{P_A>0}\mathcal{C}(P_A,R)\right).
\end{equation}
This transformation decomposes the primary problem \eqref{ct_max} into two steps: first maximizing $\mathcal{C}$ over $P_A$ constrained by $\mathcal{O}\le\epsilon$ with a fixed $R$, and then maximizing $\mathcal{T}=\mathcal{C}R$ over $R$.} In what follows, the optimization procedure is performed step by step.

\subsubsection{Optimal $P_A$}

Since both $\mathcal{C}$ and $\mathcal{O}$ increase with $P_A$, the optimal $P_A^*$ maximizing $\mathcal{C}$ is the maximal $P_A$ satisfying $\mathcal{O}(P_A) \le \epsilon$, which is $P_A^*=P_{max} \triangleq  \mathcal{O}^{-1}(\epsilon)$, where $\mathcal{O}^{-1}(\epsilon)$ denotes the inverse function of $\mathcal{O}(P_A)$. 
The following corollary develops some properties regarding $P_{max}$.
\begin{corollary}\label{corollary_opt_pa}
	The maximal transmit power $P_{max}$  is independent of the number of transmit antennas $M$, increases with the covert outage probability threshold $\epsilon$, decreases with the number of wardens $N$, and is proportional to $\lambda_J^{\alpha/2} P_J$, i.e., $P_{max}\propto\lambda_J^{\alpha/2} P_J$, where $\lambda_J$ and $P_J$ are the density of interferers and the interfering power, respectively.
\end{corollary}
\begin{proof}
Please refer to Appendix \ref{appendix_corollary_opt_pa}.
\end{proof}

{Corollary \ref{corollary_opt_pa} suggests that, in order to confront more wardens or achieve a smaller covert outage probability, a lower transmit power should be chosen. However, after introducing multiple antennas or random interference, it is feasible to guarantee the same level of covertness with a higher transmit power while reaping a reliability gain.}


\subsubsection{Optimal $R$}
Having acquired the optimal transmit power $P_A^*$, this step determines the optimal transmission rate $R^*$ that maximizes the covert throughput $\mathcal{T}=\mathcal{C}(P_A^*)R$. 
Due to the complicated expression of $\mathcal{C}$ in \eqref{oc}, it is difficult to prove the monotonicity of $\mathcal{T}$ w.r.t. $R$. However, it is intuitive that either a too large or a too small $R$ will not yield a large $\mathcal{T}$, and the optimal $R^*$ can be obtained via an exhaustive search, which is given by 
\begin{equation}\label{ct_max1}
R^* = {\arg}\max_{R>0}\mathcal{C}(P_A^*)R.
\end{equation}

{The Diophantus equation in \eqref{oc} makes $\mathcal{C}(P_A^*)$ time-consuming to calculate and complicated to analyze. In order to circumvent such a difficulty and facilitate the subsequent optimization, a practical requirement of high reliability is considered and a computational convenient suboptimal solution to problem \eqref{ct_max1} is further  provided.}
In particular, focusing on the
large connectivity probability regime $\mathcal{C}\rightarrow 1$, the following lemma gives a compact expression for $\mathcal{C}$.
\begin{lemma}\label{lemma_cp_high}
	In the large connectivity probability regime $\mathcal{C}\rightarrow 1$, $\mathcal{C}$ in \eqref{oc} is approximated by
	\begin{equation}\label{cp_high}
{	\mathcal{ C} \approx} \left(1+ K_{\alpha,M}\phi\beta^{\delta}\right)e^{-\phi\beta^{\delta}},
	\end{equation}
	where  $K_{\alpha,M}=\delta\sum_{m=1}^{M-1}({1}/{m!})
	\prod_{l=1}^{m-1}(l-\delta)$.
\end{lemma}
\begin{proof}\label{pt_high_proof}
As $\phi\rightarrow 0\Rightarrow\mathcal{C}\rightarrow 1$, discarding the high-order terms $\Theta \left(\phi^2\right)$ in \eqref{oc} yields \eqref{cp_high}.  
\end{proof}

Note that $\phi\rightarrow 0$ reflects all cases of parameters including but not limited to $r_{a,o}$, $\lambda_J$, $P_J$, and $P_A$ which may produce a sufficiently large $\mathcal{C}$.
With \eqref{cp_high}, problem \eqref{ct_max1} is recast as below:
\begin{equation}\label{ct_max2}
\max_{\beta>0}\mathcal{T}_o(\beta)\triangleq\left(1+K_{\alpha,M}\phi_o\beta^{\delta}\right)e^{-\phi_o\beta^{\delta}} \ln (1+\beta),
\end{equation}
where $\phi_o \triangleq {\kappa\lambda_J P_J^{\delta} r_{a,o}^2}/{P_{max}^{\delta}}$.
The solution to \eqref{ct_max2} is provided by the following theorem.
\begin{theorem}\label{theorem_opt_beta}
 $\mathcal{T}_o(\beta)$ in \eqref{ct_max2} is a first-increasing-then-decreasing function of $\beta$, and the optimal $\beta_o$ maximizing $\mathcal{T}_o(\beta)$ is the unique root of $\beta>0$ to the equation $Q(\beta)=0$ with $Q(\beta)$ given by
\begin{equation}\label{opt_beta}
Q(\beta) = \frac{1+K_{\alpha,M}\phi_o\beta^{\delta}}{1+\beta}-\frac{\ln(1+\beta)\left[1+K_{\alpha,M}(\phi_o\beta^{\delta}-1)\right]}{\delta^{-1}\phi_o^{-1}\beta^{1-\delta}}.
\end{equation}  
\end{theorem}
\begin{proof}
Please refer to Appendix \ref{appendix_theorem_opt_beta}.
\end{proof}

{Appendix \ref{appendix_theorem_opt_beta} shows that $Q(\beta)$ in \eqref{opt_beta} initially is positive and then becomes negative as $\beta$ grows from zero to infinity. Hence, the optimal $\beta_o$ can be efficiently calculated via a bisection method with the equation $Q(\beta)=0$, which is quite time-saving compared with the exhaustive search. More importantly,  some useful properties concerning $\beta_o$ are well preserved in the equation $Q(\beta)=0$ and can be easily extracted by the following corollary.}
\begin{corollary}\label{corollary_opt_beta}
	The optimal $\beta_o$ satisfying $Q(\beta_o)=0$ decreases with $\phi_o$ and increases with $K_{\alpha,M}$. 
\end{corollary}
\begin{proof}
The proof is completed by proving $\frac{d\beta_o}{d\phi_o}<0$ and $\frac{d\beta_o}{dK_{\alpha,M}}>0$ using the derivative rule for implicit functions with $Q(\beta_o)=0$ \cite{Zheng2015Multi}.
\end{proof}

With the aid of Corollaries \ref{corollary_opt_pa} and \ref{corollary_opt_beta}, various insights into the behavior of the optimal rate $R_o = \ln(1+\beta_o)$ and the resultant maximal covert throughput $\mathcal{T}_o=\mathcal{C}(P_A^*,\beta_o)R_o$ are developed.
\begin{proposition}\label{proposition_opt_beta_r}
The optimal rate $R_o$ and the maximal covert throughput $\mathcal{T}_o$ increase with the number of transmit antennas $M$ and the covert outage probability threshold $\epsilon$, decrease with the distance between Alice and Bob $r_{a,o}$ and the number of wardens $N$, and are invariant to the density of interferers $\lambda_J$ and the interfering power $P_J$, irrespective of $M$.
\end{proposition} 
\begin{proof}
		Please refer to Appendix \ref{appendix_proposition_opt_beta_r}.
\end{proof}

{Proposition \ref{proposition_opt_beta_r} captures an inherent contradiction between improving throughput and covertness for covert communications. Fortunately, multi-antenna techniques enable to achieve high throughput and covertness simultaneously. Moreover, the invariance property w.r.t. $\lambda_J$ and $P_J$ implies that, even facing ubiquitous interference, a superior balance between the requirements of throughput and covertness still can be struck by properly designing the transmit power.}

{
	\begin{remark}
		The above invariance property is consistent with that observed in \cite{He2017Covert} which discussed single-antenna covert communications against random interferers and a deterministic warden. This paper extends the work of \cite{He2017Covert} to the multi-antenna scenario with random wardens. Various new findings are obtained compared with \cite{He2017Covert}. In particular, the covert throughput is severely degraded when more wardens are deployed whereas can be markedly ameliorated by equipping more transmit antennas. Numerical results show that the proposed design scheme attains a significant throughput gain than that of \cite{He2017Covert}, by relaxing the constraint of reliability. 
\end{remark}
}

{
\begin{remark}
	The invariance property is valid for the interference-limited system and meanwhile the maximal transmit power $P_{max}$ can be adjusted proportionally to $\lambda_J^{\alpha/2} P_J$. If the transmit power $P_A$ is not allowed to exceed a power budget, denoted as $P_{bud}$, then for a large $\lambda_J$ or $P_J$ which yields $P_{max}>P_{bud}$, the above invariance property is compromised and $\mathcal{T}_o$ becomes decreasing with $\lambda_J$ and  $P_J$. Nevertheless, the power budget constraint is beyond the scope of this paper.
\end{remark}
}

\section{Distributed Antenna System}\label{dma}

This section optimizes the covert throughput for the DAS, where Alice deploys the $m$-th antenna at $\mathbb{L}_{a_m}$, i.e., $(r_{a_m,o},\theta_{a_m,o})$, with a transmit power $P_{D_m}$.
Since these geographically spread antennas are connected to a central processor, Alice is capable to enable them to deliver the same message simultaneously to Bob through DBF to enhance transmission reliability.

\subsection{Worst-case Covert Outage Probability}
The covert outage probability $\mathcal{O}$ for the DAS shares the same form as \eqref{oco}, where the average detection probability $\bar p_{w}$ is calculated as \eqref{average_pe}. 
The only difference lies in the power of signals received from Alice, which changes to $S_{w} = \left|\sum_{m=1}^M\sqrt{P_{D_m}}\frac{h_{a_m,o}^{\dagger}}{|h_{a_m,o}|} h_{a_m,w}r_{a_m,w}^{-\alpha/2}\right|^2$ here and is exponentially distributed with mean $\sum_{m=1}^M{P_{D_m}}r_{a_m,w}^{-\alpha}$. Finally, $\mathcal{O}$ is obtained as \eqref{oco_expression} with $\bar p_{w}$ given in \eqref{app_pe}, only requiring to revise the term $P_Ar_{a,w}^{-\alpha}$ in \eqref{app_pe} to $\sum_{m=1}^MP_{D_m}r_{a_m,w}^{-\alpha}$.

The optimal detection thresholds $\xi$ for Willies' detectors for the worst-case covert communication is obtained in a similar way for the CAS. Specifically, for the special case with $\alpha=4$, the optimal $\xi$ for Willie at $\mathbb{L}_w$ that leads to a maximal $\bar p_w$ is determined as given in Theorem \ref{theorem_opt_xi}, simply with $B$ in \eqref{opt_xi} replaced with ${1}/\left({\sum_{m=1}^MP_{D_m}r_{a_m,w}^{-4}}\right)$.

It is worth noting that the independence between $\mathcal{O}$ and $M$ for the CAS no longer holds for the DAS, and adding transmit antennas surprisingly exacerbates the covertness as will be certified in Sec. V. The reason behind is, the distributed antennas are more vulnerable to the detection by Willies who are randomly located in the network. Nevertheless, it is inferred that such deterioration will gradually vanish, and $\mathcal{O}$ will eventually approach a constant as $M$ becomes sufficiently large. 
For instance, 
it is easily verified that, the received power from an ocean of antennas arranged on a circle with transmit power $\frac{P_A}{M}$ is equal to that from a single antenna randomly distributed on the same circle with power $P_A$, i.e., $\lim_{M\rightarrow\infty}\frac{P_A}{M}\sum_{m=1}^M r_{a_m,w}^{-\alpha} =P_A \int_0^{2\pi}\frac{r^{-\alpha}_{a,w}}{2\pi}d\theta$.

\subsection{Connectivity Probability}
Revisiting the connectivity probability $\mathcal{C}$ defined in \eqref{def_oc}, $S_o$ for the DAS is given by $S_{o} = \big|\sum_{m=1}^M\sqrt{P_{D_m}}|h_{a_m,o}|r_{a_m,o}^{-\alpha/2}\big|^2$ which appears as the squared sum of independent and non-identically distributed Rayleigh random variables, rather than a gamma random variable for the CAS. This transition enormously complicates the computation of $\mathcal{C}$. Fortunately, an integral form for the exact $\mathcal{C}$ for the DAS is provided by the following theorem.
\begin{theorem}\label{theorem_cp_dma}
	The connectivity probability for the DAS is given by
 \begin{align}\label{cp_dma}
		\mathcal{C} = 1 - 
	\int_{\mathcal{V}}\frac{\psi_2e^{-\kappa\psi_0\beta^{\delta}}}{\psi_1^M}\sum_{n=1}^M\left(\delta\kappa\psi_0\beta^{\delta}\right)^n	\Upsilon_{M,n}dv_1,\cdots, dv_M,
	\end{align} 
where the domain of integration ${\mathcal{V}}$ is described as  $\{\mathcal{V}: v_1\ge 0, \cdots, v_M\ge 0, \sum_{m=1}^M v_m<1\}$, $\psi_0 = \lambda_J P_J^{\delta} \psi_1^{\delta}$, $\psi_1 = \sum_{m=1}^{M}\frac{r_{a_m,o}^{\alpha}v_m^2}{P_{D_m}}$, $\psi_2 = \prod_{m=1}^{M}\frac{2r_{a_m,o}^{\alpha}v_m}{P_{D_m}}$, and $\Upsilon_{M,n}$ is defined in Theorem \ref{theorem_cp}.
\end{theorem}
\begin{proof}
		Please refer to Appendix \ref{appendix_theorem_cp_dma}.
\end{proof}

The exactness of \eqref{cp_dma} is verified in the simulation section.
Although the monotonicity of $\mathcal{C}$ w.r.t. $P_{D_m}$ and $\beta$ is not explicitly reflected in \eqref{cp_dma} due to the multiple integral, it still can be concluded that $\mathcal{C}$ increases with $P_{D_m}$ and decreases with $\beta$ from the definition given in \eqref{def_oc}. As done for the CAS, the large connectivity probability regime where $\mathcal{C}\rightarrow 1$ is considered, followed by a more concise expression for $\mathcal{C}$ given in the following corollary. 
\begin{corollary}\label{corollary_cp_dma_high}
	In the large connectivity probability regime,  $\mathcal{C}$ in \eqref{cp_dma} is approximated by
	\begin{equation}\label{cp_dma_high}
	{\mathcal{ C} \approx }1 -W \beta^{\delta},
	\end{equation}
	where  $W \triangleq \delta\kappa\lambda_J P_J^{\delta}\Upsilon_{M,1}\int_{\mathcal{V}}\psi_1^{\delta-M}\psi_2dv_1,\cdots,dv_M$ is independent of $\beta$.
\end{corollary}
\begin{proof}
	The result follows easily by considering $\psi_0\rightarrow 0$ in \eqref{cp_dma}.  
\end{proof}

{Note that $\mathcal{C}$ in \eqref{cp_dma_high} becomes a linear decreasing function of $\lambda_J$, $P_J^{\delta}$, and $\beta^{\delta}$, respectively. This greatly simplifies the design of transmit power and  transmission rate for covert throughput maximization, as will be detailed in the next subsection. }

\subsection{Covert Throughput Maximization} 
This subsection addresses the problem of maximizing the covert throughput $\mathcal{T}=\mathcal{C}R$ subject to a covertness constraint $\mathcal{O}\le\epsilon$ for the DAS.
The problem is formulated similarly as \eqref{ct_max} and can be resolved by executing the same two-step process for the CAS described in Sec. III-D. 

{Generally, the optimal transmit power $P_{D_m}$, for $m = 1,\cdots, M$, should be jointly determined based on the distances $r_{a_m,o}$. However, this is intractable since $\mathcal{O}$ and $\mathcal{C}$ are coupled with $P_{D_m}$ in an extremely  sophisticated way. 
Motivated by the fact that Alice is absolutely unaware of the locations of both the interferrers and Willies, a plausible special case is examined, where the $M$ antennas are deployed at a same distance from Bob and with equal transmit power, i.e., $P_{D_m}=P_D$, for $m = 1,\cdots, M$.} In this case, the optimal $P_D$ that maximizes $\mathcal{C}$ for a fixed $R$ is equal to the maximal $P_D$ satisfying $\mathcal{O}\le\epsilon$, i.e., $P_D^* = \mathcal{O}^{-1}(\epsilon)$.
Subsequently, the optimal $R^*$ that maximizes $\mathcal{T}=\mathcal{C}(P_D^*)R$ can be exhaustively searched as $R^* = {\arg}\max_{R>0}\mathcal{C}(P_D^*)R$.

The following part seeks an easy-to-compute suboptimal $R_o$ by considering the large connectivity probability regime with $\mathcal{C}$ given in \eqref{cp_dma_high}. The problem is described as 
\begin{equation}\label{ct_dma}
\max_{\beta>0}\mathcal{T}_o(\beta)\triangleq\left(1-W\beta^\delta \right) \ln (1+\beta),
\end{equation}
with the solution provided by the following theorem.
\begin{theorem}\label{theorem_opt_beta_dma}
 $\mathcal{T}_o(\beta)$ in \eqref{ct_dma} first increases and then decreases with $\beta$, and reaches the maximum at $\beta=\beta_o$, where $\beta_o$ is the unique zero-crossing $\beta>0$ of the following derivative:
	\begin{equation}\label{opt_beta_dma}
\frac{d\mathcal{T}_o(\beta)}{d\beta} = \frac{1-W\beta^{\delta}}{1+\beta}-W\delta\beta^{\delta-1}\ln(1+\beta).
	\end{equation}  
\end{theorem}
\begin{proof}
	Please refer to Appendix \ref{appendix_theorem_opt_beta_dma}.
\end{proof}

As shown in Appendix \ref{appendix_theorem_opt_beta_dma}, $\frac{d\mathcal{T}(\beta)}{d\beta} $ is initially positive and then negative as $\beta$ increases from zero to $W^{-1/\delta}$, then $\beta_o$ can be rapidly searched via a bisection method with $\frac{d\mathcal{T}(\beta)}{d\beta} =0$. Moreover, invoking the derivative rule for implicit functions with  $\frac{d\mathcal{T}(\beta_o)}{d\beta_o} =0$, $\frac{d\beta_o}{dW}$ is proved to be negative. Similar to the CAS, it is proved that the maximal transmit power $P_D^*$ for the DAS is proportional to the term $\lambda_J^{\alpha/2} P_J$, i.e., $P_D^*\propto\lambda_J^{\alpha/2} P_J$. Hence, the optimal transmission rate $R^*$ and the maximal covert throughput $\mathcal{T}^*$ are invariant to $\lambda_J$ and $P_J$, regardless of the value of $M$.

\section{Simulation Results}
This section presents simulation results to verify the theoretical findings. Without loss of generality, for the CAS Alice's $M$ antennas are placed at the same location $\mathbb{L}_a=(1,0)$ with a unit distance from Bob at the origin $o$. For a fair comparison, the $M$ antennas for the DAS are arranged uniformly on the circle $\mathcal{B}(o,1)$ such that  the $m$-th antenna is located at $\mathbb{L}_{a_m}=(1,2\pi(m-1)/M)$.
Moreover, equal total transmit power is considered for the two systems and equal power allocation is assumed among the antennas for the DAS such that $P_{D_m} =P_D =  P_A/{M}$, for $m=1,\cdots,M$. Throughout the experiments, some parameters are set fixed such that $\alpha=4$ and $P_J = 30$ dBm.

\begin{figure}[!t]
	\centering
	\includegraphics[width = 3.8in]{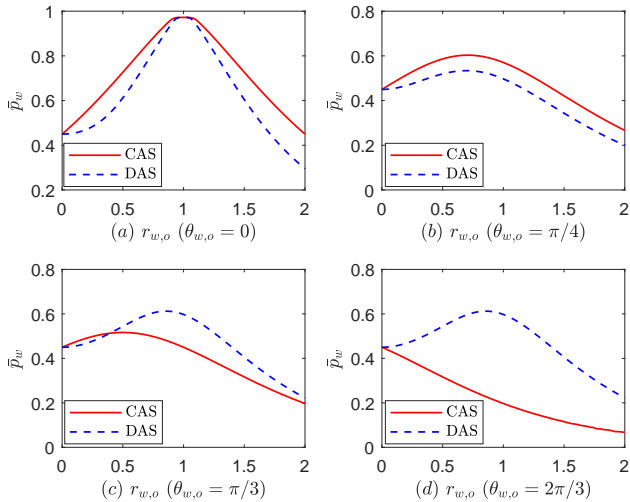}
	\caption{Average detection probability $\bar p_{w}$ vs.  Willie's location $\mathbb{L}w=(r_{w,o},\theta_{w,o})$, with {$P_A=30$ dBm}, $\lambda_J = 0.1$, and $M=4$. }
	\label{ADP_W}
\end{figure}

Fig. \ref{ADP_W} plots the average detection probability $\bar p_{w}$ of an arbitrary Willie with different locations  $\mathbb{L}_w$.
Whether the CAS or DAS is superior in terms of covertness depends heavily on Willie's location. 
Bear the locations of Alice's transmit antennas in mind, it is observed that when Willie stays closer to the co-located antennas than to the distributed ones (see the two figures above), the DAS provides a higher level of covertness (i.e., a smaller $\bar p_{w}$). Conversely, as Willie moves far away from the co-located antennas but approaches one or more distributed antennas (see the two figures below), the CAS produces a smaller $\bar p_{w}$ and is more beneficial for covert communications.

\begin{figure}[!t]
	\centering
	\includegraphics[width = 3.8in]{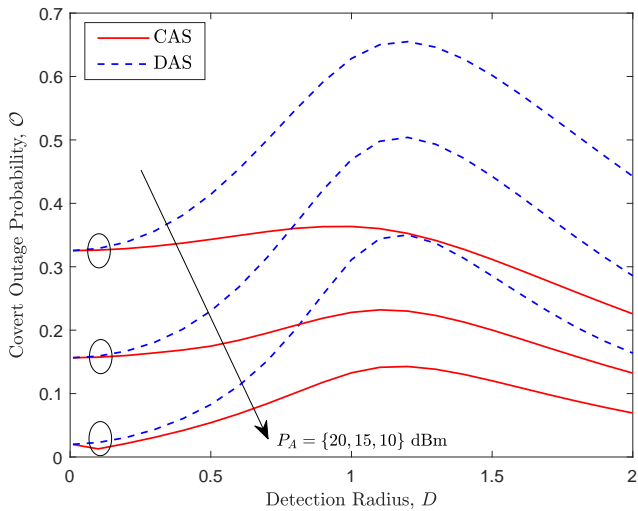}
	\caption{Covert outage probability $\mathcal{O}$ vs. detection radius $D$ for different $P_A$'s, with $\lambda_J = 0.1$, $M=4$, and $N=2$. 
}
	\label{COP_D}
\end{figure}
Fig. \ref{COP_D} depicts the covert outage probability $\mathcal{O}$ with different detection radius $D$ and transmit power $P_A$.
Although the DAS achieves a lower average detection probability than that of the CAS when Willie resides inside certain regions (as shown in Fig. \ref{ADP_W}), the CAS can invariably offer a smaller $\mathcal{O}$ compared with the DAS.
The underlying reason is that the distributed deployment of transmit antennas is a double-edged sword for covert communications. On one hand, it lowers the transmit power for each antenna, which indeed hampers Willie's detection. On the other hand, since Willie can appear anywhere in the network, the geographically scattered antennas also increase the possibility of offering Willie a larger aggregate power, which unfortunately outweighs the advantage of the reduced transmit power. Therefore, when Willies' movement is completely uninformed, the CAS is more rewarding for covert communications.
It is also observed that there exists an optimal $D^*$ leading to a maximal $\mathcal{O}$, and the value of $D^*$ is slightly larger than the distance between Alice and Bob ($r_{a,o}=1$ in this figure).
Besides, as indicated by Corollary \ref{corollary_pw_pa}, $\mathcal{O}$ increases with $P_A$ as a higher transmit power makes the detection easier.

\begin{figure}[!t]
	\centering
	\includegraphics[width = 3.6in]{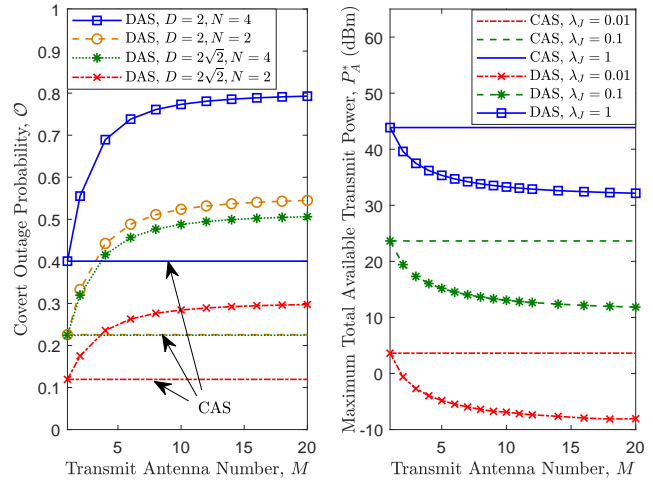}
	\caption{Left figure: $\mathcal{O}$ vs. $M$, with $\lambda_J = 0.1$ and {$P_A = 20$ dBm}; Right figure: $P_A^*$ vs. $M$, with $D = 2$, $N=2$, and $\epsilon = 0.3$. }
	\label{COP_M}
\end{figure}
Fig. \ref{COP_M} shows how multiple antennas affect the covert outage probability $\mathcal{O}$ (left figure) and Alice's maximal transmit power $P_A^*$ satisfying $\mathcal{O}\le\epsilon$ (right figure), respectively. In the left figure, $\mathcal{O}$ remains constant with $M$ for the CAS since the power perceived at Willie is equivalent to that from a single-antenna transmitter when Alice adopts MRT. For the DAS, $\mathcal{O}$ increases with $M$ and reaches a plateau for a sufficiently large $M$, as explained in Sec. IV-A. 
It is as expected that $\mathcal{O}$ decreases with the detection radius $D>1$ and increases with the warden number $N$. The two curves marked with circles and asterisks share the same density of wardens, i.e., $\frac{N}{\pi D^2}=\frac{1}{2\pi}$, but the latter gives a larger $\mathcal{O}$. This implies, narrowing down the detection range while with  fewer wardens might be more effective for detection then deploying more wardens in an expanded detection region.
The right figure reveals that $P_A^*$ achieved for the CAS is independent of $M$, whereas that for the DAS decreases with $M$.
This substantiates the disadvantage of the distributed antennas. Moreover, $P_A^*$ for both systems increases with the interferer density $\lambda_J$, which demonstrates the benefit of coexisting interferers for covert communications. 

\begin{figure}[!t]
	\centering
	\includegraphics[width = 3.8in]{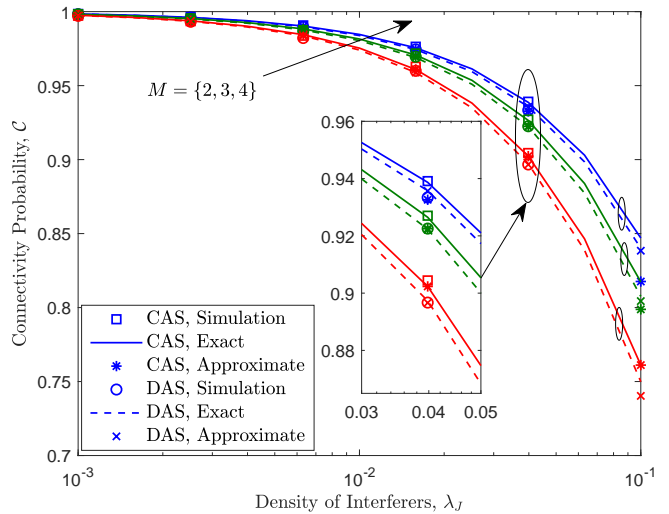}
	\caption{$\mathcal{C}$ vs. $\lambda_J$ for different $M$'s, with {$P_A = 30$ dBm}. }
	\label{CP}
\end{figure}
Fig. \ref{CP} plots the connectivity probability $\mathcal{C}$ with different densities of interferers $\lambda_J$ and numbers of transmit antennas $M$. Monte-Carlo simulation results match well with the theoretical values. The approximations derived in \eqref{cp_high} and \eqref{cp_dma_high} approach closely to the exact values in \eqref{oc} and \eqref{cp_dma} for quite a wide range of $\mathcal{C}$, respectively. {This affirms the rationality of considering the large $\mathcal{C}$ regime for designing the optimal rate. It is found that $\mathcal{C}$ increases with $M$ for both systems, and the gap between them is nearly negligible. This is because, both systems take full advantage of  the spatial degrees of freedom and enable coherent superposition of signals at the destination.}

\begin{figure}[!t]
	\centering
	\includegraphics[width = 3.8in]{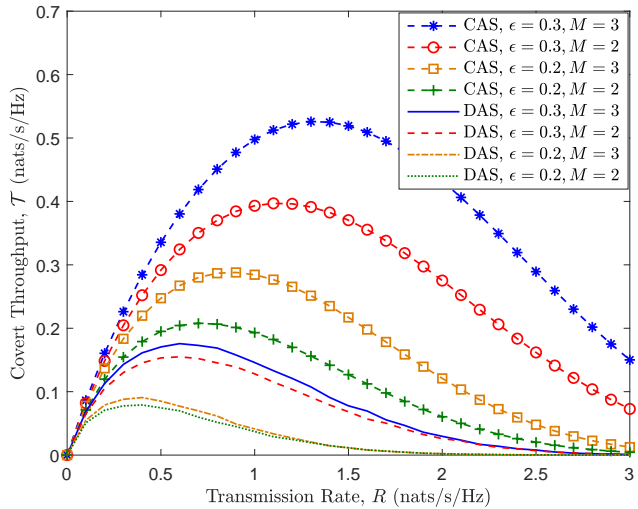}
	\caption{$\mathcal{T}$ vs. $R_s$ for different $\epsilon$'s and $M$'s, with $\lambda_J = 0.1$, $D=2$, and $N=2$. }
	\label{CT_RS}
\end{figure}
Fig. \ref{CT_RS} shows that the covert throughput $\mathcal{T}$ for both the CAS and DAS first increases and then decreases with the transmission rate $R$, as proved in Theorems \ref{theorem_opt_beta} and \ref{theorem_opt_beta_dma}. {This is because, as $R$ continues to increase, the connectivity probability $\mathcal{C}$ becomes quite small, thus leading to a low $\mathcal{T}=\mathcal{C}R$. It is found that the optimal $R$ yields a peak throughput much higher than that under a constant $R$ without optimization. This highlights the significance of the designs. }
As stated in Proposition \ref{proposition_opt_beta_r}, the optimal $R$ increases with the covert outage probability threshold $\epsilon$ and the number of transmit antennas $M$. 
{The basic cause is that, as  $\epsilon$ (also the maximal $P_A^*$) or $M$ increases, $\mathcal{C}$ rises and as a consequence a larger $R$ can be supported.}
When adding a single antenna, the CAS achieves a pronounced throughput improvement whereas the DAS only attains an insignificant gain. 
The main reason behind lies in the loss of the maximal transmit power for the DAS after guaranteeing the covertness requirement, as shown in Fig. \ref{COP_M}.

\begin{figure}[!t]
	\centering
	\includegraphics[width = 3.8in]{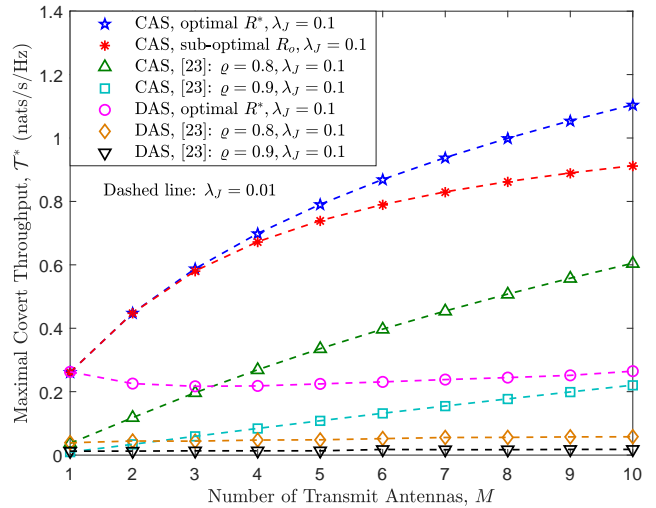}
	\caption{$\mathcal{T}^*$ vs. $M$ for different $\varrho$'s and $\lambda_J$'s, with $D=2$, $N=2$, and $\epsilon=0.3$. }
	\label{CT_M}
\end{figure}
Fig. \ref{CT_M} compares the CAS and DAS in terms of covert throughput $\mathcal{T}$ for different numbers of transmit antennas $M$. 
{The covert throughput obtained in \cite{He2017Covert} is examined as the benchmark performance in which an extra reliability requirement $\mathcal{C}\geq \varrho$ is imposed.
It is found that the throughput of the suboptimal scheme approach closely to that of the optimal scheme for comparatively few transmit antennas, and both schemes achieve a remarkable throughput gain than that from \cite{He2017Covert}.
The fundamental reason is that, by permitting a slight sacrifice of reliability, a prominent throughput improvement can be gained. 
Moreover, the CAS offers an outstanding superiority over the DAS, especially for a fairly large $M$. 
This suggests, the CAS should be a priority for covert communications in random networks. }
From the dashed lines in this figure, it is observed that $\mathcal{T}$ remains unchanged w.r.t. the density of interferers $\lambda_J$. This validates the invariance property stated in Proposition \ref{proposition_opt_beta_r}. {It is explained as follows: the concomitant interferers can effectively hinder the detection, thus enabling a higher transmit power while not compromising on the covertness requirement; the increase of transmit power in return can neutralize well the adverse impact of the interference on throughput performance.}

\section{Conclusions}
{The potential of multi-antenna systems for covert communications in random networks is examined.
Considering the worst-case covert communication where the wardens employ optimal detection thresholds for their detectors, the covert throughput is maximized by successively adjusting the transmit power and the transmission rate.
An interesting invariance property is revealed for both the CAS and DAS, which states that, whatever the number of transmit antennas is, the maximal covert throughput is not impacted by the density of interferers. 
Numerical results indicate that the CAS is more favorable to the covert communication for random networks than the DAS, particularly when a large number of transmit antennas are deployed.}

\appendix

\subsection{Proof of Theorem \ref{theorem_opt_xi}}
\label{appendix_theorem_opt_xi}
Rewrite \eqref{pe_4} as a function of $\xi$ given below:
	\begin{equation}\label{min_q}
	\bar p_w(\xi) = -{\rm erf}({A}/{\sqrt{\xi}})+e^{-B\xi}Y(\xi).
	\end{equation}
	Then, the derivative $\frac{d\bar p_w(\xi)}{d\xi}$ is given by
	\begin{equation}\label{dpe1}
	\frac{d\bar p_w(\xi)}{d\xi} = {A}e^{-\frac{A^2}{\xi}}/{\sqrt{\pi\xi^3}}-Be^{-B\xi}Y(\xi)+B{\rm erf}({A}/{\sqrt{\xi}}),
	\end{equation}
	which is due to $\frac{dY(\xi)}{d\xi} =Be^{B\xi}{\rm erf}\left({A}/{\sqrt{\xi}}\right)$. Note that $\frac{d\bar p_w(\xi)}{d\xi}|_{\xi=0} = {A}e^{-{A^2}/{\xi}}/{\sqrt{\pi\xi^3}}>0$ and $\frac{d\bar p_w(\xi)}{d\xi}|_{\xi\rightarrow\infty}\rightarrow -Be^{-B\xi}Y(\xi)<0$, which implies there is at least one zero-crossing of $\frac{d\bar p_w(\xi)}{d\xi}$. Denote an arbitrary one as $\xi_o$ such that $\frac{d\bar p_w(\xi)}{d\xi}|_{\xi=\xi_o} =0$, and substituting it into \eqref{dpe1} yields
	\begin{equation}\label{opt_xi_proof}
	Be^{-B\xi_o}Y(\xi_o)= B{\rm erf}({A}/{\sqrt{\xi_o}})+{A}e^{-\frac{A^2}{\xi_o}}/{\sqrt{\pi\xi_o^3}}.
	\end{equation}
	With \eqref{opt_xi_proof} in mind, the second derivative $\frac{d^2\bar p_w(\xi)}{d\xi^2}$ at $\xi=\xi_o$ is given by
	\begin{equation}
	\frac{d^2\bar p_w(\xi)}{d\xi^2}|_{\xi=\xi_o}={A}e^{-\frac{A^2}{\xi_o}}\left(A^2-{3}\xi_o/{2} \right)/{\sqrt{\pi\xi_o^7}}.
	\end{equation}
The sign of $\frac{d^2\bar p_w(\xi)}{d\xi^2}|_{\xi=\xi_o}$ is solely determined by the term $A^2-\frac{3}{2}\xi_o$.
The uniqueness of the zero-crossing of ${d\bar p_w(\xi)}/{d\xi}$ is proved below by contradiction. 
	
	Suppose there are $K>1$ zero-crossing points of $\frac{d\bar p_w(\xi)}{d\xi}$ sorted as $0<\xi_{o,1}\le\xi_{o,2}\le\cdots \le\xi_{o,K}$. Since $\frac{d\bar p_w(\xi)}{d\xi}|_{\xi=0} >0$, the first (i.e., the minimal) zero-crossing $\xi_{o,1}$ must satisfy $\frac{d^2\bar p_w(\xi)}{d\xi^2}|_{\xi=\xi_{o,1}}<0$, which yields $A^2-\frac{3}{2}\xi_{o,1}<0$. Otherwise, $\xi_{o,1}$ never exists.
	Evidently, for any $k>1$,  $A^2-\frac{3}{2}\xi_{o,k}<0$ and $\frac{d^2\bar p_w(\xi)}{d\xi^2}|_{\xi=\xi_{o,k}}<0$ hold as $\xi_{o,k}\ge\xi_{o,1}$. This indicates, $\frac{d\bar p_w(\xi)}{d\xi}$ initially is positive and then becomes negative after $\xi$ exceeds $\xi_{o,1}$. In other words, there is only one zero-crossing of $\frac{d\bar p_w(\xi)}{d\xi}$, denoted as $\xi_o$, which is the solution that maximizes $\bar p_w(\xi)$ and can be calculated by setting $\frac{d\bar p_w(\xi)}{d\xi}$ in \eqref{dpe1} to zero. Substituting the optimal $\xi_o$ satisfying \eqref{opt_xi_proof} into \eqref{pe_4} completes the proof.

\subsection{Proof of Corollary \ref{corollary_opt_pa}}
\label{appendix_corollary_opt_pa}
The first three properties are proved by noting that $\mathcal{O}$ remains constant with $M$, increases with $P_A$, and decreases with $N$, respectively. Since  $\mathcal{O}$ is independent of $M$, the fourth property can be proved similarly as \cite[Theorem 1]{He2017Covert} for a single-antenna transmitter. However, the proof for \cite[Theorem 1]{He2017Covert} is incomplete, since only the sufficient condition is provided. The proof is simplified here by taking the special case $\alpha = 4$ as an example. 
Recalling $\bar p_{w,max}$ in \eqref{opt_pe} 
with $A \propto {\lambda_J\sqrt{P_J}}$ and $B\propto 1/{P_A}$, it is proved that if a group $\left\{A_1,P_{A,1},\xi_{o,1}\right\}=\left\{A_0,P_{0},\xi_{0}\right\}$ yields $\bar p_{w,max}=\eta\in[0,1]$, then an arbitrary group $\left\{A_2,P_{A,2},\xi_{o,2}\right\}=\left\{\varpi A_0,\varpi^2 P_{0},\varpi^2\xi_{0}\right\}$ for $\varpi>0$ also produces $\bar p_{w,max}=\eta$.
The necessity of $P_{A,2}=\varpi^2 P_{0}$ for satisfying $\bar p_{w,max}=\eta$ when $A_2=\varpi A_0$ is further verified by the monotonicity of $\bar p_{w,max}$ w.r.t. $P_A$ shown in Corollary \ref{corollary_pw_pa}.
That means $A^2/P_A = \lambda_J^2{P_J}/P_A$ maintains invariant w.r.t. $\lambda_J$ and $P_J$ on the premise of $\bar p_{w,max}=\eta$. 
Since $\mathcal{O}$ in \eqref{oco_cal} increases with $\bar p_{w}$, the invariance property given above is also valid for $P_A=P_{max}$ with $\mathcal{O}(P_{max})=\epsilon$, i.e., $P_{max}\propto \lambda_J^2{P_J}$.

\subsection{Proof of Theorem \ref{theorem_opt_beta}}
\label{appendix_theorem_opt_beta}
The derivative of $\mathcal{T}_o(\beta)$ w.r.t. $\beta$ is calculated from \eqref{ct_max2}, which is $\frac{d\mathcal{T}_o(\beta)}{d\beta} ={e^{-\phi_o\beta^{\delta}}}Q(\beta),$
where $Q(\beta)$ is defined in \eqref{opt_beta}.
It is easy to show that $\frac{d\mathcal{T}_o(\beta)}{d\beta}|_{\beta=0}=1>0$ and $\frac{d\mathcal{T}_o(\beta)}{d\beta}|_{\beta\rightarrow\infty}<0$. Hence, at least one zero-crossing of  $\frac{d\mathcal{T}_o(\beta)}{d\beta}$ exists. 
Denote an arbitrary one as $\beta_o$ such that ${e^{-\phi_o\beta_o^{\delta}}}Q(\beta_o)=0$ which further yields $Q(\beta_o)=0$.
The second derivative $\frac{d^2\mathcal{T}_o(\beta)}{d\beta^2}$ at $\beta=\beta_o$ is computed as 
	\begin{align}\label{dt2}
	\frac{d^2\mathcal{T}_o(\beta)}{d\beta^2}|_{\beta=\beta_o}	&={e^{-\phi_o\beta_o^{\delta}}}\frac{dQ(\beta)}{d\beta}|_{\beta=\beta_o}\nonumber\\
	&=\frac{\delta e^{-\phi_o\beta_o^{\delta}}}{1+\beta_o}\left[\frac{1+K_{\alpha,M}\phi_o\beta_o^{\delta}}{\delta(1+\beta_o)}\left(1-\frac{\beta_o}{\ln(1+\beta_o)}\right)\right.\nonumber\\
	&\left.-1-\frac{ K_{\alpha,M}\phi_o\beta_o^{\delta} (1+K_{\alpha,M}\phi_o\beta_o^{\delta})}{1+K_{\alpha,M}(\phi_o\beta_o^{\delta}-1)}\right].
	\end{align}
It is known that $1+K_{\alpha,M}(\phi_o\beta_o^{\delta}-1)>0$ from
$Q(\beta_o)=0$. 
Plugging this inequality with $\beta_o>\ln(1+\beta_o)$ into \eqref{dt2} yields $\frac{d^2\mathcal{T}_o(\beta)}{d\beta^2}|_{\beta=\beta_o}<0$, which indicates that $\mathcal{T}_o(\beta)$ first increases with $\beta$ and then decreases after $\beta$ exceeds $\beta_o$. Hence, $\beta_o$ is the optimal $\beta$ that maximizes $\mathcal{T}_o(\beta)$.

\subsection{Proof of Proposition \ref{proposition_opt_beta_r}}
\label{appendix_proposition_opt_beta_r}
Note that $K_{\alpha,M}=\delta\sum_{m=1}^{M-1}({1}/{m!})
\prod_{l=1}^{m}(l-\delta)$ increases with $M$ and $\phi_o = {\kappa\lambda_J P_J^{\delta} r_{a,o}^2}/{P_{max}^{\delta}}$ decreases with $P_{max}$. As indicated by Corollary \ref{corollary_opt_pa}, $P_{max}$ increases with $\epsilon$, decreases with $N$, and is independent of $M$. Invoking Corollary \ref{corollary_opt_beta}, it is proved that $\beta_o$ increases with $M$ and $\epsilon$ and decreases with $r_{a,o}$ and $N$, respectively.
Next, as shown in Corollary \ref{corollary_opt_pa} that $P_{max}\propto\lambda_J^{\alpha/2} P_J$, it is known that $\phi_o$ does not change with $\lambda_J$ and $P_J$. Therefore, the optimal $\beta_o$ is independent of $\lambda_J$ and $P_J$. For the maximal $\mathcal{T}_o$, the monotonicity w.r.t. these parameters is the same as the monotonicity of $\mathcal{C}$ w.r.t. them. It is easy to prove that $\mathcal{C}$ increases with $M$ and $\epsilon$, decreases with $r_{a,o}$ and $N$, and remains unchanged with $\lambda_J$ and $P_J$, respectively. This completes the proof.

\subsection{Proof of Theorem \ref{theorem_cp_dma}}
\label{appendix_theorem_cp_dma}
Let $S_o =  \left(\sum_{m=1}^M X_m\right)^2$ with $X_m=\sqrt{P_{D_m}}|h_{a_m,o}|r_{a_m,o}^{-\alpha/2}$, and $\mathcal{C}$ in \eqref{def_oc} can be rewritten as 
\begin{equation}\label{cp_dma_2}
\mathcal{C}=1 - \mathbb{P}\left\{S_o<\beta I_o\right\}
=1-\mathbb{P}\left\{\sum_{m=1}^M X_m<\sqrt{{\beta I_o}}\right\}.
\end{equation}
As $X_m$ obeys the Rayleigh distribution with the PDF $f_{X_m}(x_m)=\frac{2{r_{a_m,o}^{\alpha}}{x_m}}{P_{D_m}}e^{-r_{a_m,o}^{\alpha}{x_m^2}/P_{D_m}}$, due to the mutual independence among $\{X_m\}_{m=1}^M$, the joint PDF is given by 
\begin{equation}\label{joint_pdf}
f_{X_1,\cdots,X_M}(x_1,\cdots,x_M) =\prod_{m=1}^M\frac{2{r_{a_m,o}^{\alpha}}{x_m}}{P_{D_m}} e^{-\frac{{r_{a_m,o}^{\alpha}}{x_m}}{P_{D_m}}}.
\end{equation}
Substituting \eqref{joint_pdf} into \eqref{cp_dma_2} yields
\begin{align}\label{cp_dma_3}
\mathcal{C} &=1-\mathbb{E}_{I_o}\left[\int_{\mathcal{X}}\prod_{m=1}^{M}\left(\frac{2{r_{a_m,o}^{\alpha}}{x_m}}{P_{D_m}} e^{-\frac{{r_{a_m,o}^{\alpha}}{x_m}}{P_{D_m}}}\right)dx_1,\cdots, dx_M\right]\nonumber\\
&\stackrel{\mathrm{(e)}}=
1 -  \int_{\mathcal{V}}\mathbb{E}_{I_o}\left[I_o^Me^{-\psi_1\beta I_o}\right]\beta^M\psi_2dv_1,\cdots, dv_M,
\end{align}
with the integration domain $\{\mathcal{X}: x_1\ge 0, \cdots, x_M\ge 0, \sum_{m=1}^M x_m<\sqrt{\beta I_o}\}$; $\rm (e)$ follows from the  replacement $x_m \rightarrow v_m\sqrt{\beta I_o} $. Substituting $\mathbb{E}_{I_o} \left[
I_o^me^{-sI_o}\right]=(-1)^m\frac{d^m \mathcal{L}_{I_o}(s)}{ds^m}$ in Theorem \ref{theorem_cp} with $s=\psi_1\beta$ into \eqref{cp_dma_3} completes the proof.

\subsection{Proof of Theorem \ref{theorem_opt_beta_dma}}
\label{appendix_theorem_opt_beta_dma}
From \eqref{ct_dma},  $\beta<W^{-1/\delta}$ should be ensured to achieve a positive $\mathcal{T}_o(\beta)$.
It is easily proved that $\frac{d\mathcal{T}_o(\beta)}{d\beta}$ in \eqref{opt_beta_dma} is positive at $\beta=0$ and becomes negative as $\beta\rightarrow W^{-1/\delta}$. Hence, at least one zero-crossing of $\frac{d\mathcal{T}_o(\beta)}{d\beta}$ exists. Denote an arbitrary one as $\beta_o$ such that $\frac{d\mathcal{T}_o(\beta)}{d\beta}|_{\beta=\beta_o}=0$. This yields $\frac {1-W\beta_o^{\delta}}{1+\beta_o}=W\delta\beta_o^{\delta-1}\ln(1+\beta_o)$, and then $\frac{d^2\mathcal{T}_o(\beta)}{d\beta^2}$ at $\beta=\beta_o$ is given by
\begin{equation}
\frac{d^2\mathcal{T}_o(\beta)}{d\beta^2}|_{\beta=\beta_o}=\frac{-\delta}{\beta_o(1+\beta_o)}\left[1+W\left(\frac{\beta_o-\ln(1+\beta_o)}{\beta_o^{1-\delta}}\right)\right].
\end{equation}
$\beta_o\ge\ln(1+\beta_o)$ yields $\frac{d^2\mathcal{T}_o(\beta)}{d\beta^2}|_{\beta=\beta_o}<\frac{-\delta}{\beta_o(1+\beta_o)}<0$, i.e., $\mathcal{T}_o$ is quasi-concave on $\beta$. In other words, $\mathcal{T}_o(\beta)$ initially increases and then decreases with $\beta$ and is maximized at $\beta=\beta_o$.

\end{document}